\documentclass[11pt]{article}
\usepackage{amssymb,latexsym,amsmath,enumerate,verbatim,amsfonts,amsthm}
\usepackage{cite}
\usepackage[ansinew]{inputenc}
\usepackage{graphicx,color,float,caption}
\usepackage{pgfplots}
\usepackage{tikz}
\usetikzlibrary{positioning, shapes.geometric}
\newtheorem{theorem}{Theorem}

\newtheorem{proposition}[theorem]{Proposition}
\newtheorem{lemma}{Lemma}

\begin{document}
\title{Co-positivity of tensors and bounded from below conditions of CP conserving two-Higgs-doublet potential}
\author{Yisheng Song\thanks{Corresponding author. School of Mathematical Sciences, Chongqing Normal University, Chongqing 401331 P.R. China, Email: yisheng.song@cqnu.edu.cn}}
       
\date{}
\maketitle
\abstract{ In this paper, the analytic sufficient and necessary conditions are obtained for the CP conserving two-Higgs-doublet  potential to be bounded from below  by using the co-positivity of  tensors. This is achieved by treating the potential as a  quartic homogeneous polynomial about the moduli of the two Higgs doublet fields,  where the ‘angles’ is described as the misalignment of the two doublets, then solving three minimum problems with respect to the misalignment.   Finally,  the analytic conditions are established with the help of the corresponding theory and methods of higher order tensors.}

\section{Introduction}

  The stability model of  multi-Higgs potential is very noticeable in particle physics community, and such a model itself was first proposed by Lee \cite{L1973} for the two-Higgs-doublet model (for short, 2HDM)  in 1973.  Weinberg \cite{SW} gave a general model of multi-Higgs potential  in 1976.  It was studied in hundreds of papers since then, one of the simplest extensions of the standard model is  the 2HDM. There are many studies on the bounded from below (for short, BFB)   conditions of the 2HDM potential, including CP conservation and CP violation. These results are different kinds of  analytic conditions of the BFB of such a potential, for examples, 2HDM with CP conservation in Refs. \cite{DM,BFIS,CG,I2007,K2016,K2012,K1985}, the most general 2HDM in Refs.\cite{I2007,DM}, 2HDM with CP conservation and CP violation in Ref. \cite{BFLRS,CG,K1985,N2020,NS,KKO,ERS,I2007,GK2005}, 2HDM handled numerically \cite{MMNN} and others references that are no cited here.  For the tree-level metastability bounds of the most general 2HDM see Ref.\cite{IS2015}.  Recently, Bahl et al. \cite{BCCIW} presented an analytically sufficient condition of the BFB of 2HDM with  CP violation. However, there are not a simple analytic sufficient and necessary  conditions even for 2HDM to be bounded from below until now.  In this paper, we will try our best to give a argument technique to solve this problem, and provide a simple analytical expression (Theorem \ref{thm:1}) of the bounded-from-below  for 2HDM with  CP conservation.

   It is well-known that for the 2HDM with explicit CP conservation,  all couplings of the Higgs potential are real \cite{BFLRS,IS2015,L1973,DM}. The scalar potential of such a 2HDM with Higgs doublets $\Phi_1$ and $\Phi_2$ is 
   \begin{equation}\label{eq:VH}\aligned V_H(\Phi_1,\Phi_2)=&V_2(\Phi_1,\Phi_2)+V_4(\Phi_1,\Phi_2)\\
   	V_2(\Phi_1,\Phi_2)=&m_{11}^2|\Phi_1|^2+m_{22}^2|\Phi_2|^2-m_{12}^2(\Phi_1^*\Phi_2+\Phi_2^*\Phi_1)\\
   	V_4(\Phi_1,\Phi_2)=& \Lambda_1|\Phi_1|^4+\Lambda_2|\Phi_2|^4+\Lambda_3|\Phi_1|^2|\Phi_2|^2\\
   	&+\Lambda_4(\Phi_1^*\Phi_2)(\Phi_2^*\Phi_1)+\frac{\Lambda_5}2[(\Phi_1^*\Phi_2)^2+(\Phi_2^*\Phi_1)^2]\\
   	&+\Lambda_6|\Phi_1|^2(\Phi_1^*\Phi_2+\Phi_2^*\Phi_1)\\
   	&  +\Lambda_7|\Phi_2|^2(\Phi_1^*\Phi_2+\Phi_2^*\Phi_1),\endaligned
   \end{equation}
where $\Phi^*$ is Hermitian conjugate of $\Phi$. 
The the BFB  of the scalar potential in the 2HDM  is considered only the non-negativity of the quartic part $V_4$ \cite{BFLRS},  i.e.,
$ V_4(\Phi_1,\Phi_2)\geq0.$

In this paper, with the help of the corresponding theory and methods of higher order tensors, we mainly present the sufficient and necessary conditions of the BFB for the 2HDM potential  with explicit CP conservation. That is, our main result is the following:
\begin{theorem}\label{thm:1} Let $\Lambda_1>0$, $\Lambda_2>0$. Then $V_4(\Phi_1,\Phi_2)\geq0$ if and only if
	$$\aligned (1)\ &\Lambda_6=\Lambda_7=0, \Lambda_3+2\sqrt{\Lambda_1\Lambda_2}\geq0, \Lambda_3+\Lambda_4-|\Lambda_5|+2\sqrt{\Lambda_1\Lambda_2}\geq0;\\
	(2)\	
	& \Delta\geq0,  \Lambda_3+2\sqrt{\Lambda_1\Lambda_2}\geq0, \Lambda_3+\Lambda_4-\Lambda_5+2\sqrt{\Lambda_1\Lambda_2}\geq0,\\ &|\Lambda_6\sqrt{\Lambda_2}-\Lambda_7\sqrt{\Lambda_1}|\leq2\sqrt{\Lambda_1\Lambda_2(\Lambda_3+\Lambda_4+\Lambda_5)+2\Lambda_1\Lambda_2\sqrt{\Lambda_1\Lambda_2}},\\
	&\mbox{(i) }-2\sqrt{\Lambda_1\Lambda_2}\leq \Lambda_3+\Lambda_4+\Lambda_5 \leq6\sqrt{\Lambda_1\Lambda_2}, \\
	&\mbox{(ii) }\Lambda_3+\Lambda_4+\Lambda_5>6\sqrt{\Lambda_1\Lambda_2}\mbox{ and }\\
	&|\Lambda_6\sqrt{\Lambda_2}+\Lambda_7\sqrt{\Lambda_1}| \leq 2\sqrt{\Lambda_1\Lambda_2(\Lambda_3+\Lambda_4+\Lambda_5)-2\Lambda_1\Lambda_2\sqrt{\Lambda_1\Lambda_2}},
	\endaligned$$
	where $\Delta=4(12\Lambda_1\Lambda_2-12\Lambda_6\Lambda_7+(\Lambda_3+\Lambda_4+\Lambda_5)^{2})^{3}-(72\Lambda_1\Lambda_2(\Lambda_3+\Lambda_4+\Lambda_5)+36\Lambda_6\Lambda_7(\Lambda_3+\Lambda_4+\Lambda_5)-2(\Lambda_3+\Lambda_4+\Lambda_5)^{3}-108\Lambda_1\Lambda_7^{2}-108\Lambda_6^{2}\Lambda_2)^{2}$.
\end{theorem}

\section{2HDM potential  and Real tensors}

In order to show our main result, we need turn the polynomal $V_4(\Phi_1,\Phi_2)$ about two complex variable   into  a  4th order symmetric real tensor, and then, use  some conclusions of tensors to prove our conclusion.

Let $\phi_k$=$|\Phi_k|$, the modulus of $\Phi_k$ for $k=1,2.$ Then  $$\Phi_1^*\Phi_2=\phi_1\phi_2 \rho e^{i\theta}\mbox{ and }\Phi_2^*\Phi_1=\phi_1\phi_2 \rho e^{-i\theta},$$
here $i^2=-1$ and $\rho\in[0,1]$ is the orbit space parameter \cite{K2016,GK2005}. So, we have
$$\aligned V_4(\Phi_1,\Phi_2)=& \Lambda_1\phi_1^4+\Lambda_2\phi_2^4+\Lambda_3\phi_1^2\phi_2^2+\Lambda_4\rho^2\phi_1^2\phi_2^2\\
&+\frac{\Lambda_5}2(\phi_1^2\phi_2^2\rho^2e^{2i\theta}+\phi_1^2\phi_2^2\rho^2e^{-2i\theta})\\
&+\Lambda_6\phi_1^2(\phi_1\phi_2\rho e^{i\theta}+\phi_1\phi_2\rho e^{-i\theta})\\
&+\Lambda_7\phi_2^2(\phi_1\phi_2\rho e^{i\theta}+\phi_1\phi_2\rho e^{-i\theta})\\
=& \Lambda_1\phi_1^4+\Lambda_2\phi_2^4+\Lambda_3\phi_1^2\phi_2^2+\Lambda_4\rho^2\phi_1^2\phi_2^2\\
&+\Lambda_5\rho^2\phi_1^2\phi_2^2\cos2\theta+2\Lambda_6\phi_1^3\phi_2\rho\cos\theta+2\Lambda_7\phi_2^3\phi_1\rho\cos\theta\\
=& \Lambda_1\phi_1^4+\Lambda_2\phi_2^4+(\Lambda_3+\Lambda_4\rho^2-\Lambda_5\rho^2)\phi_1^2\phi_2^2\\
&+2\Lambda_5\rho^2\phi_1^2\phi_2^2\cos^2\theta+2\Lambda_6\rho\phi_1^3\phi_2\cos\theta+2\Lambda_7\rho\phi_2^3\phi_1\cos\theta.
\endaligned$$
Let $x=\cos\theta$. Then $x\in[-1,1]$  and
\begin{equation}\label{eq:V}
	\aligned
V_4(\Phi_1,\Phi_2)=& \Lambda_1\phi_1^4+\Lambda_2\phi_2^4+(\Lambda_3+\Lambda_4\rho^2-\Lambda_5\rho^2)\phi_1^2\phi_2^2\\
&+2\Lambda_5\rho^2\phi_1^2\phi_2^2 x^2+2\rho(\Lambda_6\phi_1^2 +\Lambda_7\phi_2^2)\phi_1\phi_2 x.
\endaligned\end{equation}
This defines a 4th-order 2-dimensional symmetric tensor $\mathcal{A}(\rho,x)=(a_{ijkl})$ with a parameter $\rho$ and $x$:
\begin{equation}\label{eq:A}
\aligned
	 a_{1111}=&\Lambda_{1}, a_{2222}=\Lambda_{2}, \\
 a_{1122}=&\displaystyle\frac{1}{6}[\Lambda_{3}+\Lambda_{4}\rho^{2}+\Lambda_5\rho^2(2x^2-1)],\\
a_{1112}=&\displaystyle\frac{1}{2}\Lambda_6\rho x,\  a_{1222}=\displaystyle\frac{1}{2}\Lambda_7\rho x.
\endaligned
\end{equation}
So the the BFB of the system $V_H(\Phi_1,\Phi_2)$ may turn into the co-positivity of a 4th-order tensor $\mathcal{A}(\rho,x)$.

The positive definiteness and the co-positivity of a 4th order symmetric tensor are applied to deal with the BFB conditions of the particle physics model in Ref. \cite{K2016}. Recently, In Refs. \cite{S2021,SL2021,SQ2021,SQ2020,LS2021}, the distinctly sufficient conditions were given for the co-positivity of 4th order 3-dimensional symmetric tensors, which may be used to the BFB conditions of scalar potential of the particle physics model.

\section{Co-positivity of Matrices and Tensors}

 The co-positivity of a matrix $M=(m_{ij})$ is used to verify the BFB  of the particle physics model in Refs. \cite{K2016,K2012}.
A symmetric matrix $M=(m_{ij})$ is co-positive if the quadratic form  $x^\top Mx\geq0$ for all non-negative vectors $x\in \mathbb{R}^n$. The co-positivity of a $2\times 2$ symmetric matrix $M=(m_{ij})$ was showed in Ref. \cite[Lemma 2.1]{ACE}, also see  Hadeler \cite[Theorem 2]{H1983} and Nadler \cite[Lemma 1]{N1992} for more details. A $2\times 2$ symmetric matrix $M=(m_{ij})$ is co-positive if and only if  $$m_{11}\geq0,m_{22}\geq0 \mbox{ and } m_{12}+\sqrt{m_{11}m_{22}}\geq0.$$

The co-positivity of a symmetric tensor is tried  to test the BFB of the particle physical model in Ref. \cite{K2016}. A $m$th-order $n$-dimensional symmetric tensor $T=(t_{i_1\cdots i_m}) (i_j=1,2,\ldots,n, j=1,2,\ldots, m)$ is co-positive \cite{Q2005,Q2013,SQ2015,QCC2018,QL2017} if the $m$-degree homogeneous polynomial $T{\bf x}^m\geq0$ for all non-negative vectors ${\bf x}\in \mathbb{R}^n$, where ${\bf x}=(x_1,x_2,\cdots x_n)^\top$ and
$$T{\bf x}^m={\bf x}^\top(T{\bf x}^{m-1})=\sum_{i_1\cdots i_m=1}^nt_{i_1\cdots i_m}x_{i_1}\cdots x_{i_m},$$
$T{\bf x}^{m-1}=(y_1,y_2,\cdots, y_n)^\top$ is a vector with its entiries
$$y_k=(T{\bf x}^{m-1})_k=\sum_{i_2\cdots i_m=1}^nt_{ki_2\cdots i_m} x_{i_2}\cdots x_{i_m}.$$

 Let $f(x,y)$ be a quartic homogeneous real polynomial about two variables $x,y$,
	\begin{equation}\label{eq:f}
		f(x,y)=a_0 x^4+a_1x^3y+a_2x^2y^2+a_3xy^3+a_4y^4.
	\end{equation}
Then it gives a 4th-order 2-dimensional symmetric tensor $T=(t_{ijkl})$ with its entires,
\begin{flushleft}
	\qquad$t_{1111}=a_0,$ $t_{2222}=a_4,$  $t_{1122}=\displaystyle\frac{1}{6}a_2,$
	$t_{1112}=\displaystyle\frac{1}{4}a_1,$ $t_{1222}=\displaystyle\frac{1}{4}a_3$,
\end{flushleft}
Assume that $a_0>0$ and $a_4>0$. In Ref. \cite{SL2021}, the co-positivity of the above tensor $T$ was proved (also see Refs. \cite{UW,QYZ}).
\begin{lemma}\cite[Lemma 3.1]{SL2021}\label{lem:1} Let $a_0>0$ and $a_4>0$. Then $f(x,y)\geq0$ for all $x\geq0,\ y\geq0$ if and only if
	\begin{flushleft}
		\quad(1) $\Delta\leq0$ and $a_1\sqrt{a_4}+a_3\sqrt{a_0}>0;$ or
	\end{flushleft}
	\begin{flushleft}
		\quad(2) $a_1\geq0$, $a_3\geq0$ and $2\sqrt{a_0a_4}+a_2\geq0;$ or
	\end{flushleft}
	\begin{flushleft}
		\quad(3) $\Delta\geq0$, $|a_1\sqrt{a_4}-a_3\sqrt{a_0}|\leq4\sqrt{a_0a_2a_4+2a_0a_4\sqrt{a_0a_4}}$ and
	\end{flushleft}
	\begin{flushleft}
		\qquad(i) $-2\sqrt{a_0a_4}\leq a_2 \leq6\sqrt{a_0a_4}$,
	\end{flushleft}
	\begin{flushleft}
		\qquad(ii) $a_2>6\sqrt{a_0a_4}$ and \end{flushleft}
	\begin{flushleft}
		\qquad $a_1\sqrt{a_4}+a_3\sqrt{a_0} \geq -4\sqrt{a_0a_2a_4-2a_0a_4\sqrt{a_0a_4}}$,
	\end{flushleft}
	where $\Delta=4(12a_0a_4-3a_1a_3+a_2^{2})^{3}-(72a_0a_2a_4+9a_1a_2a_3-2a_2^{3}-27a_0a_3^{2}-27a_1^{2}a_4)^{2}$.
	\end{lemma}

\section{Bounded-from-below  conditions}

In this section, we mainly give the BFB conditions of the 2HDM  with explicit CP conservation. That is, how to logically  establish our main result, Theorem \ref{thm:1}. \\

The quartic part \eqref{eq:V} of such a CP conserving two-Higgs-doublet potential may be rewritten as follow
$$	\aligned
	V_4(\Phi_1,\Phi_2)=&2\Lambda_5\rho^2\phi_1^2\phi_2^2 x^2+2\rho(\Lambda_6\phi_1^2 +\Lambda_7\phi_2^2)\phi_1\phi_2 x\\
	&+\Lambda_1\phi_1^4+\Lambda_2\phi_2^4+(\Lambda_3+\Lambda_4\rho^2-\Lambda_5\rho^2)\phi_1^2\phi_2^2.
	\endaligned$$
If $\phi_1>0,\ \phi_2=0$ (or $\phi_1=0,\ \phi_2>0$), then $V_H(\Phi_1,\Phi_2)=\Lambda_1\phi_1^4$ (or $\Lambda_2\phi_2^4).$ Without loss of generality, we may assume $\Lambda_1>0$, $\Lambda_2>0$, $\phi_1>0,\ \phi_2>0$ in the sequel.
Let \begin{equation}\label{eq:1}
	A=2\Lambda_5\rho^2\phi_1^2\phi_2^2,\ B=2\rho(\Lambda_6\phi_1^2 +\Lambda_7\phi_2^2)\phi_1\phi_2,
\end{equation}
\begin{equation}\label{eq:2}
C=\Lambda_1\phi_1^4+\Lambda_2\phi_2^4+[\Lambda_3+(\Lambda_4-\Lambda_5)\rho^2]\phi_1^2\phi_2^2.
\end{equation}
Then, $V_4(\Phi_1,\Phi_2)$ may be seen as a quadratic function $f(x)$ about a variable $x\in[-1,1]$,
\begin{equation}\label{eq:3}
	f(x)=V_H(\Phi_1,\Phi_2)=Ax^2+Bx+C.
\end{equation}
So, if $A>0$, then as you can see from the image below, the function $f(x)$ attains its minimum $\frac{4AC-B^2}{4A}$ at a point $x=-\frac{B}{2A}\in [-1,1]$ or  attains its minimum $A+C\pm B$ at a point $x=\pm 1$ ($-\frac{B}{2A}\notin [-1,1]$). At that time, the graph-like of $f(x)$ is  as shown below:
\begin{center}\begin{tikzpicture}
	\begin{axis}[
		xlabel={$x$}, ylabel={$y$},
		xmin=-3.85, xmax=5.5,
		ymin=-0.3, ymax=2.9,
		xtick={-2, -1,  0, 1, 2},
		xticklabels={$-2$, $-1$, $0$, $1$, $2$},
		line width=1pt,
		axis lines=center,
		]
		\addplot[smooth,domain=-2:2, red!70]{(x+0.2)^2+0.2};
		\addlegendentry{\tiny $2A>B>-2A $}
		\addplot[smooth,domain=-4:2, blue!70]{(x+2.0)^2-0.2};
		\addlegendentry{\tiny $-\frac{B}{2A}<-1 $}
		\addplot[smooth,domain=-2:4]{(x-1.98)^2+0.22};
		\addlegendentry{\tiny  $-\frac{B}{2A}>1 $}
		\addplot[dashed, domain = 0:4, cyan!70] (-1,{x});
		\addplot[dashed, domain = 0:4, cyan!70] (1,{x});
	\end{axis}
\end{tikzpicture}

Figure 1: Graph of $f(x)$ ($A>0$)
\end{center}
If $A\leq0$, then  $f(x)$ attains its minimum $A-B+C$ or $A+B+C$ at a point $x=-1$ ($B\geq0$) or $x=1$  ($B<0$). At that time, the graph-like of $f(x)$ is  as shown below:
\begin{center}\begin{tikzpicture}
		\begin{axis}[
			xlabel={$x$}, ylabel={$y$},
			xmin=-2.3, xmax=2.3,
			ymin=-0.3, ymax=2.3,
			xtick={-2, -1,  0, 1, 2},
			xticklabels={$-2$, $-1$, $0$, $1$, $2$},
			line width=1pt,
			axis lines=center,
			]
			\addplot[smooth,domain=-4:2, red!70]{-0.9*(x+0.5)^2+2.2};
			\addlegendentry{\tiny $B<0 $}
			\addplot[smooth,domain=-2:4]{-(x-0.3)^2+2};
			\addlegendentry{\tiny  $B>0 $}
			\addplot[dashed, domain = 0:2.2, cyan!70] (-1,{x});
			\addplot[dashed, domain = 0:2.2, cyan!70] (1,{x});
		\end{axis}
\end{tikzpicture}

Figure 2: Graph of $f(x)$ ($A<0$)
\end{center}

So the following conclusions is easy to be obtained. 
\begin{proposition}\label{pro:1} $V_4(\Phi_1,\Phi_2)\geq0$  if and only if
$$
 \begin{cases}C\geq0,\\
A-B+C\geq0, \\
A+B+C\geq0.
   \end{cases}$$
\end{proposition}
\begin{proof}
	``Necessity".   $V_4(\Phi_1,\Phi_2)=f(x)\geq0$ implies that $$f(0)=C\geq0, f(1)=A+B+C\geq0, f(-1)=A-B+C\geq0.$$
	``Sufficiency".  The quadratic function $f(x)$ is non-negative in the interval $[-1,1]$ if and only if its minimum is non-negative in the interval $[-1,1]$. 
	The minimum of $f(x)$ is the smallest of the three numbers, $$f(-\frac{B}{2A})=\dfrac{4AC-B^2}{4A}, f(1), f(-1),$$  where $ -\frac{B}{2A}$ is the unique extreme point if $-\frac{B}{2A}\in[-1,1]$. Since the point $ -\frac{B}{2A}$ is the minimum point when $A>0$,  then $-\frac{B}{2A}\in[-1,1]$ means $$-2A\leq B\leq 2A.$$ It follows from Proposition \ref{pro:3} (2) that $B=0.$ So,  $C\geq0$ implies $f(-\frac{B}{2A})=C\geq0,$ and hence, $f(x)\geq0$ if $C\geq0$, $f(-1)\geq 0$, $f(1)\geq 0$.
\end{proof}
The following conclusions is easy to be obtained. For detail proof, see Appendix.
\begin{proposition}\label{pro:3}  (1) $C\geq0$ for all $\phi_1\geq0,\ \phi_2\geq0$  if and only if  $$  \Lambda_3+\Lambda_4-\Lambda_5+2\sqrt{\Lambda_1\Lambda_2}\geq0, \Lambda_3+2\sqrt{\Lambda_1\Lambda_2}\geq0.$$
	(2) If $-2A\leq B\leq 2A$ for all $\phi_1\geq0,\ \phi_2\geq0$, then $$\Lambda_6=0,\  \Lambda_7=0, \  \Lambda_5\geq0,\  i.e., \ B=0, A\geq0.$$
\end{proposition}

It follows from the equations \eqref{eq:1} and \eqref{eq:2} that
$$\aligned
A\pm B+C=&2\Lambda_5\rho^2\phi_1^2\phi_2^2\pm 2\rho(\Lambda_6\phi_1^2 +\Lambda_7\phi_2^2)\phi_1\phi_2\\
&+\Lambda_1\phi_1^4+\Lambda_2\phi_2^4+[\Lambda_3+(\Lambda_4-\Lambda_5)\rho^2]\phi_1^2\phi_2^2\\
=& (\Lambda_4+\Lambda_5)\phi_1^2\phi_2^2\rho^2\pm 2(\Lambda_6\phi_1^2 +\Lambda_7\phi_2^2)\phi_1\phi_2 \rho\\
&+\Lambda_1\phi_1^4+\Lambda_2\phi_2^4+\Lambda_3\phi_1^2\phi_2^2\\
=&a\rho^2\pm b\rho+c,
\endaligned$$
where
\begin{equation}\label{eq:4}
	a=(\Lambda_4+\Lambda_5)\phi_1^2\phi_2^2,\ b=2(\Lambda_6\phi_1^2 +\Lambda_7\phi_2^2)\phi_1\phi_2,
\end{equation}
\begin{equation}\label{eq:5}
	c=\Lambda_1\phi_1^4+\Lambda_2\phi_2^4+\Lambda_3\phi_1^2\phi_2^2.
\end{equation}

 For $\rho\in[0,1],$ Let
\begin{equation}\label{eq:6}
	s(\rho)=a\rho^2+b\rho+c\mbox{ and }t(\rho)=a\rho^2-b\rho+c.
\end{equation}
So, if $a>0$, then the function $s(\rho)$  reaches its minimum $\frac{4ac-b^2}{4a}$ at a point $\rho=-\frac{b}{2a}\in [0,1]$, or  attains its minimum $c$ (or  $a+c+ b$) at a point $\rho= 0$ (or $1$) ($-\frac{b}{2a}\notin [0,1]$). At that moment, the graph-like  of quadratic function $s(\rho)$ is illustrated by the following image:
\begin{center}\begin{tikzpicture}
		\begin{axis}[
			xlabel={$x$}, ylabel={$y$},
			xmin=-3, xmax=4.3,
			ymin=-0.5, ymax=3.5,
			xtick={-2, -1,  0, 1, 2},
			xticklabels={$-2$, $-1$, $0$, $1$, $2$},
			line width=1pt,
			axis lines=center,
			]
			\addplot[smooth,domain=-2:3, red!70]{(x-0.5)^2+0.2};
			\addlegendentry{\tiny $0>b>-2a $}
			\addplot[smooth,domain=-4:2, blue!70]{(x+0.6)^2+0.2};
			\addlegendentry{\tiny $-\frac{b}{2a}<0 $}
			\addplot[smooth,domain=-2:4]{(x-1.6)^2-0.2};
			\addlegendentry{\tiny  $-\frac{b}{2a}>1 $}
			\addplot[dashed, domain = 0:4, cyan!70] (0,{x});
			\addplot[dashed, domain = 0:4, cyan!70] (1,{x});
		\end{axis}
\end{tikzpicture}

Figure 3: Graph of $s(\rho)$ ($a>0$)\end{center}
 If $a\leq0$, then  $s(\rho)$ attains its minimum $c$ or $a+b+c$ at a point $\rho=0$ ($-\frac{b}{2a}\geq\frac12$) or $\rho=1$  ($-\frac{b}{2a}\leq\frac12$); if $a=0$, then  $s(\rho)$ attains its minimum $c$ or $b+c$ at a point $\rho=0$ ($b\geq0$)  or $\rho=1$ ($b\leq0$). The graph-like of $s(\rho)$ is  as shown below:
 \begin{center}\begin{tikzpicture}
 		\begin{axis}[
 			xlabel={$x$}, ylabel={$y$},
 			xmin=-2.1, xmax=3,
 			ymin=-0.3, ymax=2.3,
 			xtick={-2, -1,  0, 0.5, 1, 2},
 			xticklabels={$-2$, $-1$, $0$, $\frac12$, $1$, $2$},
 			line width=1pt,
 			axis lines=center,
 			]
 			\addplot[smooth,domain=-4:3, red!70]{-(x-0.89)^2+1.9};
 			\addlegendentry{\tiny $-\frac{b}{2a}>\frac12 $}
 			 			\addplot[smooth,domain=-3:4]{-(x-0.2)^2+1.9};
 			\addlegendentry{\tiny  $-\frac{b}{2a}<\frac12 $}
 		\addplot[dashed, domain = 0:2.1, blue!70] (0.5,{x});
 		\addlegendentry{\tiny  $x=\frac12 $}
 		\addplot[dashed, domain = -0.05:2.1, cyan!70] (0,{x});
 		\addplot[dashed, domain = -0.01:2.1, cyan!70] (1,{x});
 		\end{axis}
 \end{tikzpicture}

Figure 4: Graph of $s(\rho)$ ($a<0$)\end{center}

Similarly, if $a>0$, then the function $t(\rho)$ reaches its minimum $\frac{4ac-b^2}{4a}$ at a point $\frac{b}{2a}\in [0,1]$, or  attains its minimum $c$ (or  $a+c+ b$) at a point $\rho= 0$ (or $1$) ($\frac{b}{2a}\notin [0,1]$). The graph-like of $t(\rho)$ is the image below:
\begin{center}\begin{tikzpicture}
		\begin{axis}[
			xlabel={$x$}, ylabel={$y$},
			xmin=-2.5, xmax=4.3,
			ymin=-0.5, ymax=3.5,
			xtick={-2, -1,  0, 1, 2},
			xticklabels={$-2$, $-1$, $0$, $1$, $2$},
			line width=1pt,
			axis lines=center,
			]
			\addplot[smooth,domain=-2:3, red!70]{(x-0.5)^2+0.2};
			\addlegendentry{\tiny $0<b<2a $}
			\addplot[smooth,domain=-4:2, blue!70]{(x+0.6)^2-0.2};
			\addlegendentry{\tiny $\frac{b}{2a}<0 $}
			\addplot[smooth,domain=-2:4]{(x-1.6)^2+0.2};
			\addlegendentry{\tiny  $\frac{b}{2a}>1$}
			\addplot[dashed, domain = 0:3.2, cyan!70] (0,{x});
			\addplot[dashed, domain = 0:3.3, cyan!70] (1,{x});
		\end{axis}
\end{tikzpicture}

Figure 5: Graph of $t(\rho)$ ($a>0$)\end{center}
 if $a<0$, then  $t(\rho)$ attains its minimum $c$ or $a-b+c$ at a point $\rho=0$ ($\frac{b}{2a}\geq\frac12$) or $\rho=1$  ($\frac{b}{2a}\leq\frac12$); if $a=0$, then  $t(\rho)$ attains its minimum $c$ or $-b+c$ at a point $\rho=0$ ($b\leq0$)  or $\rho=1$ ($b\geq0$). The graph-like of $t(\rho)$ is the following: 
 \begin{center}\begin{tikzpicture}
 		\begin{axis}[
 			xlabel={$x$}, ylabel={$y$},
 			xmin=-1.5, xmax=2.5,
 			ymin=-0.3, ymax=2.3,
 			xtick={ -1,  0, 0.5, 1, 2},
 			xticklabels={ $-1$, $0$, $\frac12$, $1$, $2$},
 			line width=1pt,
 			axis lines=center,
 			]
 			\addplot[smooth,domain=-4:3, red!70]{-(x-0.89)^2+1.9};
 			\addlegendentry{\tiny $\frac{b}{2a}>\frac12 $}
 			\addplot[smooth,domain=-3:4]{-(x-0.2)^2+1.9};
 			\addlegendentry{\tiny  $\frac{b}{2a}<\frac12 $}
 			\addplot[dashed, domain = 0:2, blue!70] (0.5,{x});
 			\addlegendentry{\tiny  $x=\frac12 $}
 			\addplot[dashed, domain = -0.1:2, cyan!70] (0,{x});
 			\addplot[dashed, domain = 0:2, cyan!70] (1,{x});
 		\end{axis}
 \end{tikzpicture}

Figure 6: Graph of $t(\rho)$ ($a<0$)\end{center}

\begin{proposition}\label{pro:4}
(1)	$A+B+C\geq0$  if and only if
	$$
				c\geq0\mbox{ and }
			a+b+c\geq0.$$
(2) $A-B+C\geq0$   if and only if
$$
		c\geq0\mbox{ and }
		a-b+c\geq0.	
$$
\end{proposition}
\begin{proof}
(1)	``Necessity".   $A+B+C=s(\rho)\geq0$ ($\rho\in[0,1]$) implies that $$s(0)=c\geq0, s(1)=a+b+c\geq0.$$
	``Sufficiency".  The quadratic function $s(\rho)$ is non-negative in the interval $[0,1]$ if and only if its minimum is non-negative in the interval $[0,1]$. 
	The minimum of $s(\rho)$ is the smallest of the three numbers, $$s(-\frac{b}{2a})=\dfrac{4ac-b^2}{4a}, s(0), s(1),$$  where $ -\frac{b}{2a}$ is the unique extreme point if $-\frac{b}{2a}\in[0,1]$. Since the point $ -\frac{b}{2a}$ is the minimum point when $a>0$,  then $-\frac{b}{2a}\in[0,1]$ means $$-2a\leq b\leq 0.$$ It follows from Proposition \ref{pro:5} (2) that $b=0.$ So,  $c\geq0$ implies $s(-\frac{b}{2a})=c\geq0,$ and hence, $s(\rho)\geq0$ if $c\geq0$,  $s(1)\geq 0$. 
	
Using the same proof, it is easy to show (2).
\end{proof}
So the following conclusions is easy to be obtained. For detail proof, see Appendix.
\begin{proposition}\label{pro:5} (1) $c\geq0$  if and only if $$\Lambda_1\geq0, \Lambda_2\geq0,\Lambda_3+2\sqrt{\Lambda_1\Lambda_2}\geq0.$$
(2)	If  $-2a\leq b\leq 0$, then $\Lambda_6=0, \Lambda_7=0, \Lambda_4+\Lambda_5\geq0$, i.e., $$b=0,a\geq0.$$
(3)	If $2a\geq b\geq 0$, then $\Lambda_6=0, \Lambda_7=0, \Lambda_4+\Lambda_5\geq0$, i.e., $$b=0,a\geq0.$$
\end{proposition}

Next we show  $a\pm b+c\geq0$ for all $\phi_1\geq0,\ \phi_2\geq0$.
\begin{proposition}\label{pro:7}
$a-b+c\geq0$ for all $\phi_1\geq0,\ \phi_2\geq0$ if and only if \begin{flushleft} (1) $\Delta\leq0$, $\Lambda_6\sqrt{\Lambda_2}+\Lambda_7\sqrt{\Lambda_1}<0$;  \end{flushleft} \begin{flushleft} (2) $\Lambda_6\leq0$, $\Lambda_7\leq0$, $\Lambda_3+\Lambda_4+\Lambda_5+2\sqrt{\Lambda_1\Lambda_2}\geq0$; \end{flushleft}
\begin{flushleft}  (3) $\Delta\geq0$,
$|\Lambda_6\sqrt{\Lambda_2}-\Lambda_7\sqrt{\Lambda_1}|\leq2\sqrt{\Lambda_1\Lambda_2(\Lambda_3+\Lambda_4+\Lambda_5)+2\Lambda_1\Lambda_2\sqrt{\Lambda_1\Lambda_2}}$ and
\end{flushleft}
\begin{flushleft}
\qquad(i) $-2\sqrt{\Lambda_1\Lambda_2}\leq \Lambda_3+\Lambda_4+\Lambda_5 \leq6\sqrt{\Lambda_1\Lambda_2}$,
\end{flushleft}
\begin{flushleft}
\qquad(ii) $\Lambda_3+\Lambda_4+\Lambda_5>6\sqrt{\Lambda_1\Lambda_2}$ and\end{flushleft}
\begin{flushleft}
\qquad$\Lambda_6\sqrt{\Lambda_2}+\Lambda_7\sqrt{\Lambda_1} \leq 2\sqrt{\Lambda_1\Lambda_2(\Lambda_3+\Lambda_4+\Lambda_5)-2\Lambda_1\Lambda_2\sqrt{\Lambda_1\Lambda_2}}$,
\end{flushleft}
where $\Delta=4(12\Lambda_1\Lambda_2-12\Lambda_6\Lambda_7+(\Lambda_3+\Lambda_4+\Lambda_5)^{2})^{3}-(72\Lambda_1\Lambda_2(\Lambda_3+\Lambda_4+\Lambda_5)+36\Lambda_6\Lambda_7(\Lambda_3+\Lambda_4+\Lambda_5)-2(\Lambda_3+\Lambda_4+\Lambda_5)^{3}-108\Lambda_1\Lambda_7^{2}-108\Lambda_6^{2}\Lambda_2)^{2}$.
\end{proposition}
\begin{proof}
From the equations \eqref{eq:4} and \eqref{eq:5}, it follows that
$$\aligned
a-b+c=&\Lambda_1\phi_1^4+\Lambda_2\phi_2^4+(\Lambda_3+\Lambda_4+\Lambda_5)\phi_1^2\phi_2^2\\
      &-2\Lambda_6\phi_1^3\phi_2 -2\Lambda_7\phi_2^3\phi_1\\
      =& f(\phi_1,\phi_2)=a_0 \phi_1^4+a_1\phi_1^3\phi_2+a_2\phi_1^2\phi_2^2+a_3\phi_1\phi_2^3+a_4\phi_2^4,
\endaligned$$
where $a_0=\Lambda_1$, $a_1=-2\Lambda_6$, $a_2=\Lambda_3+\Lambda_4+\Lambda_5$, $a_3=-2\Lambda_7$, $a_4=\Lambda_2$. Then  an application of the co-positivity of quartic form \eqref{eq:f} (Lemma \ref{lem:1}) yields that the sufficient and necessary conditions of the inequality $a-b+c\geq0$ is \begin{flushleft} (1) $\Delta\leq0$, $-2\Lambda_6\sqrt{\Lambda_2}-2\Lambda_7\sqrt{\Lambda_2}>0$;  \end{flushleft} \begin{flushleft} (2) $-2\Lambda_6\geq0$, $-2\Lambda_7\geq0$, $\Lambda_3+\Lambda_4+\Lambda_5+2\sqrt{\Lambda_1\Lambda_2}\geq0$; \end{flushleft}
\begin{flushleft}  (3) $\Delta\geq0$,
$|2\Lambda_6\sqrt{\Lambda_2}-2\Lambda_7\sqrt{\Lambda_1}|\leq4\sqrt{\Lambda_1\Lambda_2(\Lambda_3+\Lambda_4+\Lambda_5)+2\Lambda_1\Lambda_2\sqrt{\Lambda_1\Lambda_2}}$ and
\end{flushleft}
\begin{flushleft}
\qquad(i) $-2\sqrt{\Lambda_1\Lambda_2}\leq \Lambda_3+\Lambda_4+\Lambda_5 \leq6\sqrt{\Lambda_1\Lambda_2}$,
\end{flushleft}
\begin{flushleft}
\qquad(ii) $\Lambda_3+\Lambda_4+\Lambda_5>6\sqrt{\Lambda_1\Lambda_2}$ and $-2\Lambda_6\sqrt{\Lambda_2}-2\Lambda_7\sqrt{\Lambda_1} \geq -4\sqrt{\Lambda_1\Lambda_2(\Lambda_3+\Lambda_4+\Lambda_5)-2\Lambda_1\Lambda_2\sqrt{\Lambda_1\Lambda_2}}$.
\end{flushleft}
The required conclusions follow.
\end{proof}

Similarly, the following conclusion is easy to be showed.

\begin{proposition}\label{pro:8}
$a+b+c\geq0$ for all $\phi_1\geq0,\ \phi_2\geq0$ if and only if \begin{flushleft} (1) $\Delta\leq0$, $\Lambda_6\sqrt{\Lambda_2}+\Lambda_7\sqrt{\Lambda_1}>0$;  \end{flushleft} \begin{flushleft} (2) $\Lambda_6\geq0$, $\Lambda_7\geq0$, $\Lambda_3+\Lambda_4+\Lambda_5+2\sqrt{\Lambda_1\Lambda_2}\geq0$; \end{flushleft}
\begin{flushleft}  (3) $\Delta\geq0$,
$|\Lambda_6\sqrt{\Lambda_2}-\Lambda_7\sqrt{\Lambda_1}|\leq2\sqrt{\Lambda_1\Lambda_2(\Lambda_3+\Lambda_4+\Lambda_5)+2\Lambda_1\Lambda_2\sqrt{\Lambda_1\Lambda_2}}$ and
\end{flushleft}
\begin{flushleft}
\qquad(i) $-2\sqrt{\Lambda_1\Lambda_2}\leq \Lambda_3+\Lambda_4+\Lambda_5 \leq6\sqrt{\Lambda_1\Lambda_2}$,
\end{flushleft}
\begin{flushleft}
\qquad(ii) $\Lambda_3+\Lambda_4+\Lambda_5>6\sqrt{\Lambda_1\Lambda_2}$ and \end{flushleft}
\begin{flushleft}
\qquad$\Lambda_6\sqrt{\Lambda_2}+\Lambda_7\sqrt{\Lambda_1} \geq -2\sqrt{\Lambda_1\Lambda_2(\Lambda_3+\Lambda_4+\Lambda_5)-2\Lambda_1\Lambda_2\sqrt{\Lambda_1\Lambda_2}}$.
\end{flushleft}
\end{proposition}

Combining Propositions \ref{pro:4} (1)  with Propositions \ref{pro:5} (1) and \ref{pro:8}, the following results are easy to obtain.

\begin{proposition}\label{pro:9}
	 $A+B+C\geq0$ for all $\phi_1\geq0,\ \phi_2\geq0$ if and only if  $\Lambda_3+2\sqrt{\Lambda_1\Lambda_2}\geq0$ and
	 $$\aligned (1)\  &\Delta\leq0, \Lambda_6\sqrt{\Lambda_2}+\Lambda_7\sqrt{\Lambda_1}>0; \\
	(2)\ &\Lambda_6\geq0, \Lambda_7\geq0, \Lambda_3+ \Lambda_4+\Lambda_5+2\sqrt{\Lambda_1\Lambda_2}\geq0;\\
(3)\	&  \Delta\geq0,\\ &|\Lambda_6\sqrt{\Lambda_2}-\Lambda_7\sqrt{\Lambda_1}|\leq2\sqrt{\Lambda_1\Lambda_2(\Lambda_3+\Lambda_4+\Lambda_5)+2\Lambda_1\Lambda_2\sqrt{\Lambda_1\Lambda_2}},\\
	 &\mbox{(i) }-2\sqrt{\Lambda_1\Lambda_2}\leq \Lambda_3+\Lambda_4+\Lambda_5 \leq6\sqrt{\Lambda_1\Lambda_2}, \\
	 &\mbox{(ii) }\Lambda_3+\Lambda_4+\Lambda_5>6\sqrt{\Lambda_1\Lambda_2}\mbox{ and }\\
	 &\Lambda_6\sqrt{\Lambda_2}+\Lambda_7\sqrt{\Lambda_1} \geq- 2\sqrt{\Lambda_1\Lambda_2(\Lambda_3+\Lambda_4+\Lambda_5)-2\Lambda_1\Lambda_2\sqrt{\Lambda_1\Lambda_2}}.
	 \endaligned$$
\end{proposition}

Simultaneously applying Proposition \ref{pro:5} (1) and Proposition \ref{pro:7} to Propositions \ref{pro:4} (2), the following results are easy to obtain. 
\begin{proposition}\label{pro:10} 	$A-B+C\geq0$ for all $\phi_1\geq0,\ \phi_2\geq0$ if and only if $\Lambda_3+2\sqrt{\Lambda_1\Lambda_2}\geq0$ and
	$$\aligned (1)\ & \Delta\leq0, \Lambda_6\sqrt{\Lambda_2}+\Lambda_7\sqrt{\Lambda_1}<0;\\
	(2)\ &\Lambda_6\leq0, \Lambda_7\leq0, \Lambda_3+\Lambda_4+\Lambda_5+2\sqrt{\Lambda_1\Lambda_2}\geq0;\\
	(3)\	& \Delta\geq0,\\ &|\Lambda_6\sqrt{\Lambda_2}-\Lambda_7\sqrt{\Lambda_1}|\leq2\sqrt{\Lambda_1\Lambda_2(\Lambda_3+\Lambda_4+\Lambda_5)+2\Lambda_1\Lambda_2\sqrt{\Lambda_1\Lambda_2}},\\
	&\mbox{(i) }-2\sqrt{\Lambda_1\Lambda_2}\leq \Lambda_3+\Lambda_4+\Lambda_5 \leq6\sqrt{\Lambda_1\Lambda_2}, \\
	&\mbox{(ii) }\Lambda_3+\Lambda_4+\Lambda_5>6\sqrt{\Lambda_1\Lambda_2}\mbox{ and }\\
	&\Lambda_6\sqrt{\Lambda_2}+\Lambda_7\sqrt{\Lambda_1} \leq 2\sqrt{\Lambda_1\Lambda_2(\Lambda_3+\Lambda_4+\Lambda_5)-2\Lambda_1\Lambda_2\sqrt{\Lambda_1\Lambda_2}}.
	\endaligned$$
\end{proposition}

By combining Propositions \ref{pro:1}  with Propositions \ref{pro:3}, \ref{pro:9} and \ref{pro:10}, we easily showed our main result, Theorem \ref{thm:1}.

\begin{proof}[\bf The proof of Theorem \ref{thm:1}]It follows from Propositions \ref{pro:1} and \ref{pro:3} that $V_4(\Phi_1,\Phi_2)\geq0$  if and only if
\begin{equation}\begin{aligned}  &\Lambda_3+2\sqrt{\Lambda_1\Lambda_2}\geq0,\  \Lambda_3+\Lambda_4-\Lambda_5+2\sqrt{\Lambda_1\Lambda_2}\geq0, \\ 
&A-B+C\geq0,  \\
&A+B+C\geq0.  
\end{aligned}\end{equation}
Propositions \ref{pro:9} (2) and  \ref{pro:10} (2) mean that $$\Lambda_6=\Lambda_7=0,  \Lambda_3+2\sqrt{\Lambda_1\Lambda_2}\geq0, \Lambda_3+\Lambda_4+\Lambda_5+2\sqrt{\Lambda_1\Lambda_2}\geq0,$$
The above several inequalities together  is equivalent to 
$$\Lambda_6=\Lambda_7=0,  \Lambda_3+2\sqrt{\Lambda_1\Lambda_2}\geq0, \Lambda_3+\Lambda_4-|\Lambda_5|+2\sqrt{\Lambda_1\Lambda_2}\geq0.$$
This obtain (1) of Theorem \ref{thm:1}.

The two inqualities $A-B+C\geq0$ and $A+B+C\geq0$ imply that there is contradiction between Proposition \ref{pro:9} (1) and Proposition \ref{pro:10} (1), and so, it can not hold. 

 By Propositions \ref{pro:9} (3) and \ref{pro:10} (3),  (2) of  Theorem \ref{thm:1} follows easily. 
\end{proof}

\section{Conclusions}
In summary, we prove the analytical sufficient and necessary conditions of the BFB  conditions of 2HDM potential with explicit CP conservation. At the same time, for a 4th-order 2-dimensional symmetric tensor $\mathcal{A}(\rho,x)=(a_{ijkl})$ \eqref{eq:A} with two parameters $\rho\in[0,1]$ and $x\in[-1,1]$, the co-positivity is proved. 
\section{Some remarks}

     \ \ \ \ \ \ (1) When $\Lambda_6 = \Lambda_7 = 0$,  Theorem \ref{thm:1} coincides with the corresponding conclusion of Refs. \cite{DM,GK2005,K1985,N2020,NS,KKO,ERS,BFLRS}. That is,  $V_4^{\mathbb{Z}_2}(\Phi_1,\Phi_2)\geq0$  if and only if  $$\Lambda_1>0, \Lambda_2>0, \Lambda_3+2\sqrt{\Lambda_1\Lambda_2}\geq0, \Lambda_3+\Lambda_4-|\Lambda_5|+2\sqrt{\Lambda_1\Lambda_2}\geq0.$$ 

(2) In terms of the eigenvalues of a $4\times4$ matrix,  Ivanov \cite{I2007} gave the sufficient and necessary conditions of  the BFB of CP conserving 2HDM potential,  but not analytic.   It is not clear how to do this analytically, and this may be slow to implement numerically. Our conclusions is analytic for the CP conserving potential, which  are easy to code up, and presumably faster to run, for the large scans of the 2HDM parameter space which often take place in the literature.  \\

(3) For the BFB of CP conserving 2HDM potential, Kannike \cite{K2016}  presented a sufficient condition (Eq. (88)) by means of Lagrange multiplier, which is not completely analytic.  Recently, Bahl et al. \cite{BCCIW} gave a stronger sufficient conditions (Eqs. (5.20) and (5.21)) for the BFB of CP conserving 2HDM potential using Ulrich and Watson’s result (Eq.(30))\cite{UW}.  In this paper, we use the optimizing version (Lemma 1) of Ulrich and Watson’s result to yield the analytic expression of the BFB (For the more details  about Lemma 1 see Refs. \cite{SL2021,QYZ}).  For example,  $$\aligned	& \Lambda_1=\Lambda_2=1, \Lambda_3=-1, \Lambda_4=2,  \Lambda_5=-1, \Lambda_6=1, \Lambda_7=-1.\mbox{ Then, }\\
&\Lambda_3+2\sqrt{\Lambda_1\Lambda_2}=-1+2>0,\\& \Lambda_3+\Lambda_4-\Lambda_5+2\sqrt{\Lambda_1\Lambda_2}=-1+2-(-1)+2>0,  \\
& \Lambda_3+\Lambda_4+\Lambda_5+2\sqrt{\Lambda_1\Lambda_2}=-1+2-1+2>0,\\
&|\Lambda_6\sqrt{\Lambda_2}-\Lambda_7\sqrt{\Lambda_1}|=2<2\sqrt{\Lambda_1\Lambda_2(\Lambda_3+\Lambda_4+\Lambda_5)+2\Lambda_1\Lambda_2\sqrt{\Lambda_1\Lambda_2}}=2\sqrt2,\\
&-2\sqrt{\Lambda_1\Lambda_2}=-2< \Lambda_3+\Lambda_4+\Lambda_5=0 <6\sqrt{\Lambda_1\Lambda_2}=6, \\ &\Delta=4(12\Lambda_1\Lambda_2-12\Lambda_6\Lambda_7+(\Lambda_3+\Lambda_4+\Lambda_5)^{2})^{3}-(72\Lambda_1\Lambda_2(\Lambda_3+\Lambda_4+\Lambda_5)\\
&+36\Lambda_6\Lambda_7(\Lambda_3+\Lambda_4+\Lambda_5)-2(\Lambda_3+\Lambda_4+\Lambda_5)^{3}-108\Lambda_1\Lambda_7^{2}-108\Lambda_6^{2}\Lambda_2)^{2}\\&=4(12+12+0)^3-(0-108-108)^2>0.
\endaligned $$
That is, these parameters meet Theorem \ref{thm:1} (2), which means $V_4(\Phi_1,\Phi_2)\geq0$.  However, they can't satisfy the condition Eq. (5.20) of Bahl et.al.\cite{BCCIW}, i.e.
$$\aligned	& 3\sqrt{\Lambda_1\Lambda_2}-(\Lambda_3+|\Lambda_4|+|\Lambda_5|)=3-2>0,\\ &\sqrt{\Lambda_1\Lambda_2}+\Lambda_3-(|\Lambda_4|+|\Lambda_5|+4\left|\Lambda_6\sqrt[4]{\dfrac{\Lambda_2}{\Lambda_1}}\right|)=1-1-(2+1+4)<0,\\
&\sqrt{\Lambda_1\Lambda_2}+\Lambda_3-(|\Lambda_4|+|\Lambda_5|+4\left|\Lambda_7\sqrt[4]{\dfrac{\Lambda_1}{\Lambda_2}}\right|)=1-1-(2+1+4)<0.
\endaligned $$

(4)  The quartic part of such a 2HDM potential is the following \begin{equation}\label{eq:9}\aligned V_4(\Phi_1,\Phi_2
	=&\Lambda_1(\Phi_1^*\Phi_1)^2+\Lambda_2(\Phi_2^*\Phi_2)^2\\&+\Lambda_3(\Phi_1^*\Phi_1)(\Phi_2^*\Phi_2)+\Lambda_4(\Phi_1^*\Phi_2)(\Phi_2^*\Phi_1)\\
	&+\frac{\Lambda_5}2(\Phi_1^*\Phi_2)^2+\frac{\Lambda_5^*}2(\Phi_2^*\Phi_1)^2\\&+(\Phi_1^*\Phi_1)(\Lambda_6\Phi_1^*\Phi_2+\Lambda_6^*\Phi_2^*\Phi_1)\\
	&+(\Phi_2^*\Phi_2)(\Lambda_7\Phi_1^*\Phi_2+\Lambda_7^*\Phi_2^*\Phi_1)\\
	=& \Lambda_1\phi_1^4+\Lambda_2\phi_2^4\\&+(\Lambda_3+\Lambda_4\rho^2+\rho^2(\textbf{Re}\Lambda_5\cos2\theta+\textbf{Im}\Lambda_5\sin2\theta))\phi_1^2\phi_2^2\\
		&+2\rho(\textbf{Re}\Lambda_6\cos\theta+\textbf{Im}\Lambda_6\sin\theta)\phi_1^3\phi_2 \\&+2\rho(\textbf{Re}\Lambda_7\cos\theta+\textbf{Im}\Lambda_7\sin\theta)\phi_2^3\phi_1.
		\endaligned\end{equation}
Let $x=\cos\theta$	and $y=\sin\theta$.  Then the general  2HDM potential have three parameters about `angles', $(\rho, x, y)$ with $ x^2+y^2=1$.  So  the argument approach of Theorem \ref{thm:1} may not be applied directly to the general  2HDM potential. How to revised this approach, which deserves further study and perfectness.

%
\section*{Data Availablity Statement}
This manuscript has no associated data or the data will not be deposited. [Authors' comment: This is a theoretical study and there are no external data associated with the manuscript.]
\section*{Acknowledgements} The authors would like to express their sincere thanks to the editors and anonymous referees for his/her constructive comments and valuable suggestions, and to Professors Igor P. Ivanov, K. G. Klimenko, Garv Chauhan for useful discussions and for reading  the manuscript. This work was supported by the National Natural Science Foundation of P.R. China (Grant No. 12171064), by The team project of innovation leading talent in Chongqing (No.CQYC20210309536) and by the Foundation of Chongqing Normal University (20XLB009). 

\section*{Appendix}
 
\begin{proof}[\bf The proof of Propositions \ref{pro:3}] (1) $C\geq0,$ i.e. $$C=\Lambda_1\phi_1^4+\Lambda_2\phi_2^4+[\Lambda_3+(\Lambda_4-\Lambda_5)\rho^2]\phi_1^2\phi_2^2\geq0.$$
	which is equivalent to the co-positivity of $2\times 2$ matrix $M=(m_{ij})$ with its entries,
	$$m_{11}=\Lambda_1,\ m_{22}=\Lambda_2,\ m_{12}=m_{21}=\frac12[\Lambda_3+(\Lambda_4-\Lambda_5)\rho^2].$$
	That is, $\Lambda_3+(\Lambda_4-\Lambda_5)\rho^2+2\sqrt{\Lambda_1\Lambda_2}\geq0.$ Therefore, this is equivalent to
	$$\aligned \Lambda_3+\Lambda_4-\Lambda_5+2\sqrt{\Lambda_1\Lambda_2}\geq0,\\
	\Lambda_3+2\sqrt{\Lambda_1\Lambda_2}\geq0.\endaligned$$ (2) From the equation \eqref{eq:1}, it follows that
	$$\aligned
	B-2A=2(\Lambda_6\phi_1^2 +\Lambda_7\phi_2^2-2\Lambda_5\rho\phi_1\phi_2)\rho\phi_1\phi_2,\\
	B+2A=2(\Lambda_6\phi_1^2 +\Lambda_7\phi_2^2+2\Lambda_5\rho\phi_1\phi_2)\rho\phi_1\phi_2,
	\endaligned$$
	and hence,
	$$\aligned B-2A\leq0\Leftrightarrow -\Lambda_6\phi_1^2 -\Lambda_7\phi_2^2+2\Lambda_5\rho\phi_1\phi_2\geq0,\\
	B+2A\geq0\Leftrightarrow\ \ \Lambda_6\phi_1^2 +\Lambda_7\phi_2^2+2\Lambda_5\rho\phi_1\phi_2\geq0.\endaligned$$
	Thus,  $B-2A\leq0$  is equivalent to the co-positivity of a matrix $M=(m_{ij})$ with its entries $m_{11}=-\Lambda_6, m_{22}=-\Lambda_7,  m_{12}=m_{21}=\Lambda_5\rho.$ From the co-positivity  of $2\times 2$ matrix, it follows that $$ B-2A\leq0 \Leftrightarrow \Lambda_6\leq0,\Lambda_7\leq0 \mbox{ and } \Lambda_5\rho+\sqrt{\Lambda_6\Lambda_7}\geq0.$$
	
	Similarly, we also have $$B+2A\geq0 \Leftrightarrow \Lambda_6\geq0,\Lambda_7\geq0 \mbox{ and } \Lambda_5\rho+\sqrt{\Lambda_6\Lambda_7}\geq0.$$ Thus, the inequalities $-2A\leq B\leq 2A$ imply that  $ \Lambda_6=0,\Lambda_7=0, \Lambda_5\geq0$, i.e., $B=0, A\geq0$. 
	 The required conclusion follows.
\end{proof}

\begin{proof}[\bf The proof of Propositions \ref{pro:5}] (1) It is obvious that $$c=\Lambda_1\phi_1^4+\Lambda_2\phi_2^4+\Lambda_3\phi_1^2\phi_2^2\geq0,$$
	is equivalent to
	$$\Lambda_1\geq0, \Lambda_2\geq0,\Lambda_3+2\sqrt{\Lambda_1\Lambda_2}\geq0.$$ 
	(2) It follows from the equations \eqref{eq:4} that
	$$b+2a=2(\Lambda_6\phi_1^2 +\Lambda_7\phi_2^2+(\Lambda_4+\Lambda_5)\phi_1\phi_2)\phi_1\phi_2,$$
	$$2a-b=2(-\Lambda_6\phi_1^2 -\Lambda_7\phi_2^2+(\Lambda_4+\Lambda_5)\phi_1\phi_2)\phi_1\phi_2,$$
	and hence, $$\aligned b+2a\geq0\Leftrightarrow& \Lambda_6\phi_1^2 +\Lambda_7\phi_2^2+(\Lambda_4+\Lambda_5)\phi_1\phi_2\geq0,\\
	b\leq0\Leftrightarrow& \Lambda_6\phi_1^2 +\Lambda_7\phi_2^2\leq0.
	\endaligned$$
	So, we have
	$$\aligned b+2a\geq0\Leftrightarrow& \Lambda_6\geq0, \Lambda_7\geq0, (\Lambda_4+\Lambda_5)+2\sqrt{\Lambda_6\Lambda_7}\geq0,\\
	b\leq0\Leftrightarrow& \Lambda_6\leq0, \Lambda_7\leq0,
	\endaligned$$
	and hence, $ \Lambda_4+\Lambda_5\geq0$, $\Lambda_6=0, \Lambda_7=0$, i.e., $b=0, a\geq0$. 
	
	Similarly, (3) is also  showed easily.
\end{proof}

	\end{document}